%% file: conference_101719.tex
\documentclass[conference]{IEEEtran}
\IEEEoverridecommandlockouts
\usepackage[dvips]{graphicx}
\usepackage{amsmath}
\usepackage{dsfont}
\usepackage{algorithmic}
\usepackage{array}
\usepackage{fixltx2e}
\usepackage{stfloats}
\usepackage{graphicx}
\usepackage{romannum}
\usepackage{gensymb}
\usepackage{cite}
\usepackage{amsmath,amssymb,amsfonts}
\usepackage{algorithmic}
\usepackage{graphicx}  
\usepackage{subcaption} % for subfigures
\usepackage{textcomp}
\usepackage{xcolor}
\usepackage{gensymb}
\usepackage{lettrine}
\usepackage{mathtools}  
\usepackage{xfrac}  
\usepackage[english]{babel}
\usepackage{amsmath,amsthm}
\usepackage[shortconst]{physconst}

\newtheorem{corollary}{Corollary}
\newtheorem{lemma}{Lemma}

\theoremstyle{definition}
\newtheorem{remark}{Remark}
\newtheorem{note}{Note}
\newtheorem{definition}{Definition}

\newcommand{\MK}[1]{{\color{black} {#1}}}

\def\BibTeX{{\rm B\kern-.05em{\sc i\kern-.025em b}\kern-.08em
    T\kern-.1667em\lower.7ex\hbox{E}\kern-.125emX}}

\begin{document}

\title{Spectrum Coexistence of Satellite-borne Passive Radiometry and Terrestrial Next-G Networks
}

\author{ 
\IEEEauthorblockN{Mohammad Koosha and Nicholas Mastronarde}
\IEEEauthorblockA{\textit{University at Buffalo, Department of Electrical Engineering}\\ Email: \{mkoosha, nmastron\}@buffalo.edu}
\thanks{The work of M. Koosha and N. Mastronarde was supported in part by the NSF under Award \#2030157.}
% \thanks{979-8-3503-9978-3/22/\$31.00 \copyright2022 IEEE}
}
\maketitle

\begin{abstract}
Spectrum coexistence between terrestrial Next-G cellular networks and space-borne remote sensing (RS) is now gaining attention. One major question is how this would impact RS equipment. In this study, we develop a framework based on stochastic geometry to evaluate the statistical characteristics of radio frequency interference (RFI)
%, such as its average, variance, skewness, kurtosis, and its higher central moments, 
originating from a large-scale terrestrial Next-G network operating in the same frequency band as an RS satellite. For illustration, we consider a network operating in the restricted L-band (1400-1427 MHz) with NASA's Soil Moisture Active Passive (SMAP) satellite, which is one of the latest RS satellites active in this band. We use the Thomas Cluster Process (TCP) to model RFI from clusters of  cellular base stations on SMAP's antenna's main- and side-lobes. We show that a large number of active clusters can operate in the restricted L-band without compromising SMAP's mission if they avoid interfering with the main-lobe of its antenna. This is possible thanks to SMAP's extremely low side-lobe antenna gains.
%Using our model, we show that a large number of active clusters can operate in the restricted L-band owing to SMAP's extremely low side-lobe antenna gains, without compromising SMAP's mission. 
%One question that has to be answered in this course

\end{abstract}


\section{Introduction}
The spectrum crunch has fostered extensive research on the coexistence of different wireless technologies within the same spectrum. One case that has been gaining attention in recent years involves the use of passive Radio Frequency (RF) bands, which are solely devoted to passive sensing applications such as remote sensing and radio astronomy, for active wireless communications. Specifically, the coexistence of terrestrial active wireless communications and Earth Exploration Satellite Services (EESS) is becoming a central topic of discussion \cite{polese2021coexistence}. A major question that needs to be answered is how and to what extent Radio Frequency Interference (RFI) would impact such EESS satellites.

% In this study, we develop a mathematical framework to model the RFI originating from a large-scale terrestrial cellular network and its impact on an EESS satellite. Specifically, we imagine clusters of cellular base stations exposed to an EESS satellite, where each cluster has a number of base stations active in the same frequency band as the satellite. To account for the randomness of the position of the clusters on Earth and the number of active cells within a cluster, we use the Thomas Cluster Process (TCP) from stochastic geometry \cite{afshang2016modeling}. 

\MK{While current research primarily examines spectrum sharing between terrestrial cellular networks and terrestrial passive sensing technologies, such as the Radio Dynamic Zones (RDZ) proposed by \cite{zheleva2023radio} and the Shared Spectrum Access Zones (SSAZ) for the coexistence of cellular wireless communications and terrestrial radio astronomy systems proposed by \cite{ramadan2017new},} in this study, we develop a mathematical framework to model the RFI originating from a large-scale terrestrial cellular network and its impact on an EESS satellite. Specifically, we imagine clusters of cellular base stations exposed to an EESS satellite, where each cluster has a number of base stations active in the same frequency band as the satellite. To account for the randomness of the position of the clusters on Earth and the number of active cells within a cluster, we use the Thomas Cluster Process (TCP) from stochastic geometry \cite{afshang2016modeling}.  

For illustration, we develop our model based on the National Aeronautics and Space Administration (NASA) Soil Moisture Active Satellite (SMAP) \cite{entekhabi2014smap}, which is one of the latest RS satellites active in the restricted L-band ($1400-1427$ MHz). We develop the characteristic function of RFI at both the main- and side-lobes of SMAP's antenna. Using the characteristic function, we then derive the statistical properties of the RFI, namely, its average, variance, 
%skewness, and kurtosis.
\MK{and higher central moments}
We also demonstrate that, due to the very low side-lobe gains of SMAP's antenna, a large number of terrestrial clusters can be active while exposed to SMAP's side-lobe without compromising the accuracy of SMAP's measurements. 

\MK{The paper is organized as follows. In Section \ref{sec:Preliminaries} we introduce SMAP's measurement mechanism, our methodology for RFI analysis, and other preliminaries. Section \ref{sec:RFI_Analysis} is the main RFI analysis section. Section \ref{sec:Results} is the simulations and results section. Section \ref{sec:conclusion} is the conclusion section.}

% \begin{notation}
%     We denote random variables by upper-case letters and their realizations or other deterministic quantities by lower-case letters. We use bold font to denote vectors and normal font to denote scalar quantities.  For instance, $X$ and $\bold{X}$ denote a one-dimensional (scalar) random variable and a vector random variable, respectively. Similarly, $x$ and $\mathtt{x}$ denote a scalar and a vector, respectively.
% \end{notation}

\section{Preliminaries} \label{sec:Preliminaries}
\subsection{SMAP \& Brightness Temperature}
As depicted in Figure~\ref{fig:smap_characteristics}, SMAP has a 6-meter-wide conically-scanning golden mesh reflector with a 3-dB antenna beam-width of $2.4{\degree}$ that projects a footprint of roughly $40 \times 40$ km$^2$ \MK{with an incident angle of $40^{\degree}$ at an altitude of $685$ $km$}. An Ortho-Mode Transducer (OMT) feedhorn collects the reflected radiations from the mesh reflector and duplexes them separately into \textit{vertical} and \textit{horizontal} polarizations. Figure~\ref{fig:SAG} shows a 2-dimensional cut of SMAP's antenna gain for the vertical polarization.
Through sectorization, which is a common method in stochastic geometry, we represent SMAP's antenna gain for each polarization $(p)$ as:
\begin{equation}
g=\left\{
\begin{array}{ll}
    g_{(ml)} , & \mbox{if } |d| \leq 1.2{\degree},  \\
    g_{(sl)}, & \mbox{if } |d| > 1.2{\degree},
\end{array}    
\right. \label{eq:G_SMAP}
\end{equation}
where $(ml)$ and $(sl)$ stand for \textit{main-lobe} and \textit{side-lobe}, respectively, and $|d|$ is the deviation from the main-lobe axis. For each polarization $(p)$, SMAP separately captures the \textit{brightness temperature} of soil, $t_{soil}^{(p)}$ \MK{(in Kelvin)}, from the antenna footprint by capturing the soil's natural passive thermal radiations. These brightness temperature measurements can be translated to soil moisture content using models like the Tau-Omega mode \cite{entekhabi2014smap}. We use the Nyquist noise formula \cite{turner2012johnson} to convert electromagnetic power to brightness temperature as:
\begin{equation}
t^{(p)}=\frac{\mathrm{p}^{(p)}}{k_{b}\beta},
\end{equation}
where $\mathrm{p}^{(p)}$ is the electromagnetic power received by polarization $(p)$, $k_b$ is the \textit{Boltzmann} constant, and $\beta$ is the radio frequency \textit{bandwidth}. 

\begin{note}
    Due to the symmetrical nature of SMAP's antenna gain for both polarizations, as well as the symmetry in the RFI scenario, we assume identical RFI characteristics for both of SMAP's polarizations. Thus, the discussions that follow hold true for SMAP's measurements in both polarizations.
\end{note}
\begin{figure}
\centering
  \includegraphics[width=0.95\linewidth]{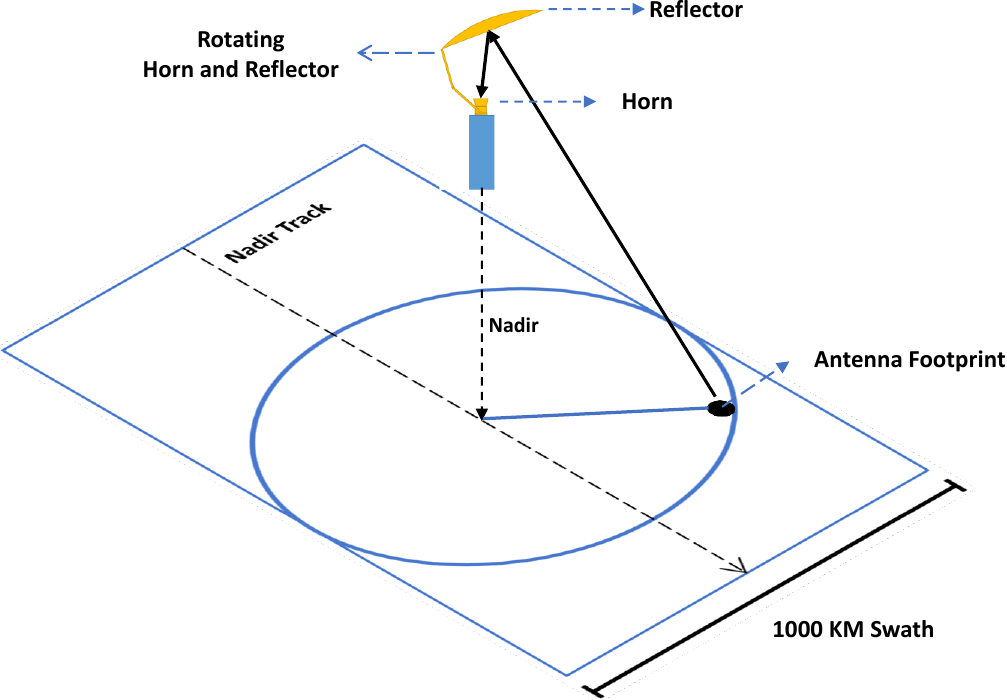}
  \caption{Horn and reflector rotation, beam footprint of the reflector, and SMAP's nadir track.}
  \label{fig:smap_characteristics}
\end{figure} 

\begin{figure}
\centering
  \includegraphics[width=0.95\linewidth]{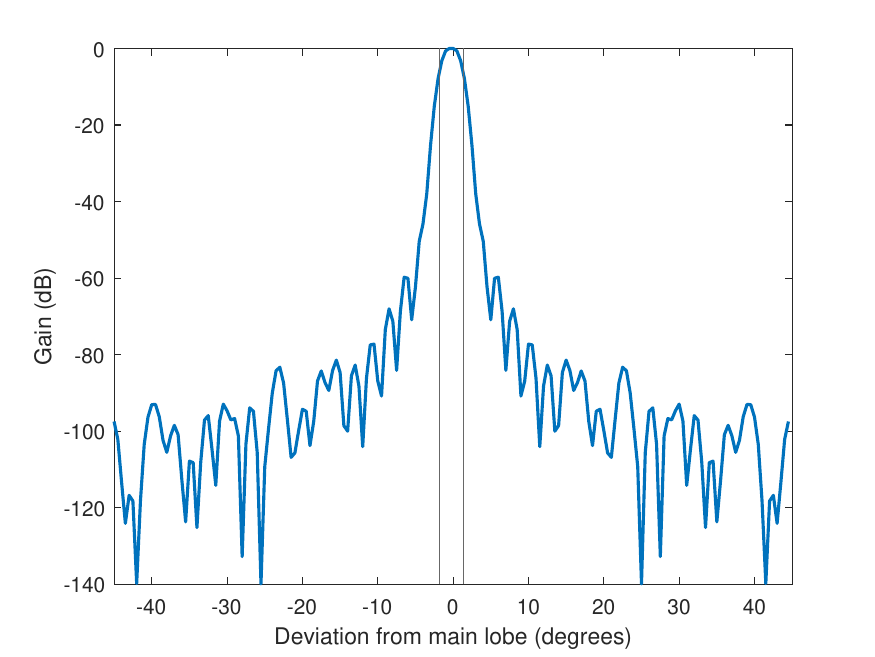}
  \caption{A 2-dimensional cut of SMAP's conical antenna gain for the vertical polarization. The gain for the horizontal polarization is similar. The two vertical lines show the $2.4^{\degree}$ antenna beam-width.}
  \label{fig:SAG}
\end{figure}

\subsection{Methodology for RFI Analysis}
In this section, we provide a concise explanation of the underlying logic guiding our analysis of RFI on SMAP. SMAP's measurements for each polarization can be seen as:
\begin{equation}
    t_{meas}=t_{soil}+T_{RFI} \label{eq:T_meas},
\end{equation}
where $T_{RFI}$ denotes the RFI temperature at SMAP. Since $T_{RFI}$ is a random variable, it causes uncertainty in SMAP's measurements. According to SMAP's documentation, uncertainties below threshold value $\tau = 1.3K$ are acceptable for SMAP's measurements \cite{entekhabi2014smap}. To model $T_{RFI}$ from a large terrestrial network, we imagine a set of $\{i\}$ Base Station (BS) clusters on the Earth-cut exposed to SMAP, where each cluster $i$ comprises a set of $\{j\}_i$ BSs, each BS$_{ij}$ has a (maximum) total electromagnetic transmission power $p_{tx}$, and $P_{ij}$ is the amount of power received by SMAP from BS$_{ij}$. Clusters $\{i\}_{(ml)} \subset \{i\}$ are located on SMAP's main-lobe antenna footprint and clusters $\{i\}_{(sl)} \subset \{i\}$ are exposed to SMAP's side-lobe. Accordingly, $T_{RFI}$ in \eqref{eq:T_meas}, can be decomposed into its main- and side-lobe components as:
\begin{align}
    &T_{RFI} = T_{(ml)}+T_{(sl)}, \label{eq:T_RFI}
\end{align}
where 
\begin{align}
    &T_{(l)}=\sum\nolimits_{\{i\}_{(l)}}\sum\nolimits_{\{j\}_i} T_{ij} \label{eq:T_RFI_l} \\
    \text{with} \quad &T_{ij}=\frac{P_{ij}}{k_{b}\beta} \label{eq:P_ij_to_T_ij}
\end{align}
where $(l)$ is either $(ml)$ or $(sl)$. Assessing each $P_{ij}$ would be a complicated function of BS$_{ij}$ antenna angles, obstructions in the environment, and SMAP's elevation angle relative to BS$_{ij}$. However, based on Free Space Path Loss (FSPL), we note that:
\begin{equation}
    P_{ij} \in \left[0,\; g\left(\frac{c}{4\pi f d_{ij}} \right)^{\alpha}p_{tx}\right]  \label{eq:Max_P_ij}  
\end{equation}
and therefore, according to \eqref{eq:P_ij_to_T_ij}, 
\begin{equation}
    T_{ij} \in \left[0,\; g\eta \omega^{\alpha} d_{ij}^{-\alpha}\right] \label{eq:T_ij_range},
\end{equation}
where $\eta=\frac{p_{tx}}{k_{b}\beta}$, $\omega=\left( \frac{c}{4\pi f} \right)$, $g$ is the gain of SMAP's antenna (based on \eqref{eq:G_SMAP}, $g = g_{(ml)}$ if $i \in \{i\}_{(ml)}$, and $g = g_{(sl)}$ if $i \in \{i\}_{(sl)}$), $c$ is the  speed of light, $f$ is the frequency, $d_{ij}$ is the distance of BS$_{ij}$ to SMAP, and $\alpha>2$ is the path loss exponent. In \eqref{eq:Max_P_ij}, we ignore atmospheric loss since it is proven to have a negligible effect in the L-band \cite{entekhabi2014smap}; however, this means that our model slightly overestimates the RFI.
%also ignoring it only means overestimating RFI. 
% We recognize that transient phenomena such as fading and shadowing can lead to temporary spikes in electromagnetic power. As outlined in SMAP's documentation, the full-band integration time for SMAP's sampling ADC is $300\micro s$, while the sub-band integration time is approximately $1.2 ms$. These integration periods highlight the importance of accounting for fast-fading effects. Although it is possible to employ fast-fading models to establish an upper limit for $P_{ij}$ within an acceptable confidence interval, we have chosen to omit this consideration in our current study, deferring it for future research. Instead, to assess maximum RFI on SMAP we take $T_{ij}:=g\eta \omega^{\alpha} d_{ij}^{-\alpha}$ in \eqref{eq:T_RFI_l}.

\MK{We acknowledge that transient phenomena like fading and shadowing can cause temporary spikes in electromagnetic power. SMAP's documentation indicates a full-band integration time of $300~\micro s$ and a sub-band integration time of approximately $1.2~ms$, emphasizing the need to address fast-fading effects. While fast-fading models could establish an upper limit for $P_{ij}$ with confidence, we have opted to omit this in our current study for future research. Instead, for assessing maximum RFI on SMAP, we use $T_{ij}:=g\eta \omega^{\alpha} d_{ij}^{-\alpha}$ in \eqref{eq:T_RFI_l}. This is akin to assuming the worst-case scenario, where all the transmission power from a single base station is directed solely towards the satellite.}

\subsection{Geometric Assumptions}
As depicted in Figure~\ref{fig:smap_exposed}, the Earth's center is the \textit{origin} $(0,0,0)$ and SMAP is located at the point $\mathtt{h}=(0,0,h)$, where $h$ is the distance of SMAP from the Earth's center. 
\begin{figure}
\centering
  \includegraphics[width=0.95\linewidth]{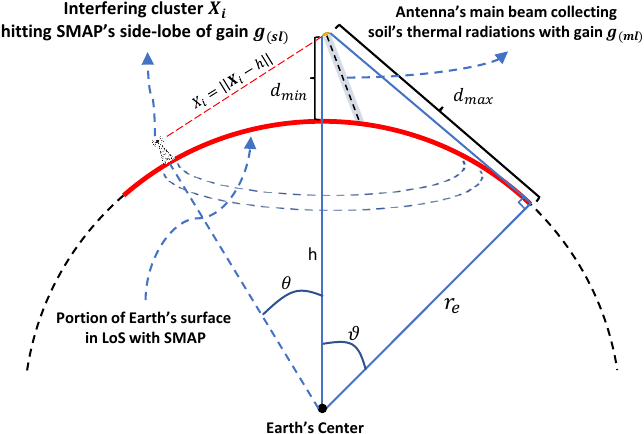}
  \caption{2-dimensional representation of the Earth-cut exposed to the satellite. The red curve represents the maximum extent of the Earth-cut exposed.}
  \label{fig:smap_exposed}
\end{figure}The area exposed to the satellite, shown with the red cap 
 and encircled in $\theta \in [0, \vartheta]$ in Figure~\ref{fig:smap_exposed}, is defined with the \textit{Borel-set} $\mathcal{B}=\left \{\lVert \mathtt{x} \lVert = r_e, \; \frac{\mathtt{x} \cdot \mathtt{h}}{\lVert \mathtt{x} \lVert \lVert \mathtt{h} \lVert} \geq \cos{(\vartheta)}  \right\}$ in $\mathbb{R}^{3}$ measure space, where $r_e$ is the Earth radius, $\lVert \cdot \lVert$ is the Euclidean norm, and $\cos{(\vartheta)}=\frac{r_e}{h}$. We define $\mathcal{B}_{(ml)} \subset \mathcal{B}$ as SMAP's antenna footprint projected on Earth 
 %where the subscript $(ml)$ stands for \textit{main-lobe} of SMAP's antenna.
 and $\mathcal{B}_{(sl)}=\mathcal{B} \big\backslash \mathcal{B}_{(ml)}$ as the set of area exposed to the side-lobe of SMAP's antenna. 
 
 Let $\Psi = \{\bold{X}_i, i \in \mathbb{N}\} \subset \mathcal{B}$ denote a parent \textit{Poisson Point Process} (PPP) with intensity measure $\Lambda(d\mathtt{x})$, where $\bold{X}_i$ denotes a cluster center. For simplicity we use a \textit{homogeneous} PPP such that $\Lambda(d\mathtt{x})=\lambda_c$. \MK{According to Figure \ref{fig:smap_exposed}, the minimum distance of a cluster center to SMAP is $d_{min}=h-r$, while the maximum distance is $d_{max}=\sqrt{h^2-r_e^2}$.} With $\psi=\{\mathtt{x}_i\}$ defined as a realization of $\Psi$, for each point $\mathtt{x} \in \psi$, we associate an \textit{independent and identically distributed} (i.i.d) offspring PPP $\Xi_\mathtt{x}$. Each cluster $\Xi_\mathtt{x}$ consists of $N$ i.i.d random points (BSs) where, according to a \textit{Thomas Cluster Process} (TCP), $N \sim \textsf{Pois}{(\lambda_{BS})}$. We define $\Phi \sim  \mathcal{P} (\lambda_c, \lambda_{BS})$, where
\begin{equation}
    \Phi=\bigcup\nolimits_{\bold{x} \in \Psi} \left(\bold{x} + \Xi_{\bold{x}} \right).
\end{equation}
Accordingly, the cluster centers located in SMAP's main- and side-lobes are respectively $\Psi_{(ml)}=\Psi(\mathcal{B}_{(ml)})$ and $\Psi_{(sl)} = \Psi \big \backslash \Psi_{(ml)}$, and their respective clusters are $\Phi_{(ml)}$ and $\Phi_{(sl)}$.

\begin{note}
    Defining $X_i= \lVert \bold{X}_i - \mathtt{h} \lVert$ as the distance of terrestrial cluster center $\bold{X_i}$ to SMAP and $x_i$ as its realization, we assume that the cluster's dispersion $\ll x_i$, since an urban area is on the order of a few kilometers, while the cluster distance to SMAP is more than $d_{min}=685$ km. 
    Thus, for simplicity, we assume that the off-spring (BSs) within each cluster are equidistant (with distance $x_i$) to SMAP.
    %Thus, for simplicity, we assume equidistant $x_i$ off-spring (BSs) within each cluster to the satellite.
\end{note}

\section{RFI Analysis} \label{sec:RFI_Analysis}
  % \vspace{-2mm}
To assess the RFI brightness temperature $T_{RFI}$ on SMAP as defined in \eqref{eq:T_RFI}, we analyze the statistical properties of RFI on SMAP's main- and side-lobes, i.e., its average, variance, and \MK{higher central moments} (skewness and kurtosis). For this purpose, we use the concepts of Cumulants and Cumulant Generating Functions (CGFs).

  \vspace{-2mm}
\begin{definition}
The CGF of random variable $X$ is defined as:
\begin{equation}
    K(t) = \log \mathds{E}[e^{tX}],
\end{equation}
which is the $\log$ of the \textit{Moment Generating Function} (MGF) $M(t) = \mathds{E}[e^{tX}]$ of random variable $X$. Accordingly the $n$th \textit{cumulant} of $X$ is as follows:
\begin{equation}
    k_{n} = K^n(0),
\end{equation}
where $K^n(t)$ denotes the $n$th derivative of $K(t)$.
\end{definition} 

  \vspace{-3mm}
\begin{remark}
    For random variable $X$, $k_1=\mu_0$, where $\mu_0$ is the first moment of $X$, i.e., $\mathds{E}[X]$. For $n \in \{ 2,3\}$, $k_n=\mu_n$, where $\mu_n$ is the $n$th central moment of $X$. Consequently, the variance of $X$ is its $2$nd central moment, i.e., $k_2 = \mu_2$. Lastly, higher order central moments $\mu_{n}$ for $n>3$ can be acquired by a combination of cumulants of $X$. For example, $\mu_4=k_4+3(k_2)^2$.
\end{remark}
Based on Definition 1 and Remark 1, we shift our focus on finding the MGF of $T_{(ml)}$ and $T_{(sl)}$ defined in \eqref{eq:T_RFI}. For this purpose, we start with the MGF of RFI brightness temperature of one cluster of BSs. 

  \vspace{-1mm}
\subsection{MGF for one Cluster}
One key quantity that can help us determine the total RFI brightness temperature at SMAP is the maximum RFI brightness temperature contributed by \textit{one} cluster. Assuming the cluster is located at point $\mathtt{x}\in \psi$ and comprises $N$ BSs that are equidistant to SMAP, we have:
  \vspace{-2mm}
\begin{equation}
    T_{cluster}(\mathtt{x}) = g  \eta \omega^{\alpha} \lVert \mathtt{x} -\mathtt{\MK{h}} \lVert^{-\alpha} N. \label{eq:T_cluster}
\end{equation}

  \vspace{-2mm}
\begin{lemma}
    For a cluster located at point $\mathtt{x}\in \psi$ with $N \sim \textsf{Pois}(\lambda_{BS})$ equidistant BSs (with distance $x=\lVert \mathtt{x}-\mathtt{h}\lVert$) to the satellite, the MGF of \eqref{eq:T_cluster} is as follows:
    \begin{multline}
    M_{cluster}(t;\; x, \;g) = \\ \exp \left( 
        \lambda_{BS}\left(-1+\exp \left( g  \eta \omega^{\alpha} x^{-\alpha} t \right) \right)
    \right). \label{eq:MGF_Cluster}
    \end{multline}
\end{lemma}
\begin{proof} Since $N$ is a Poisson random variable, the MGF of $N$ is $M_{N}(t)=\exp \left(\lambda_{BS}(e^t-1) \right)$. By setting $t:=g\eta \omega^{\alpha}x^{-\alpha} t$, we acquire \eqref{eq:MGF_Cluster}.
\end{proof}

$T_{\text{cluster}}(\mathtt{x})$ in \eqref{eq:T_cluster} is the fundamental unit of RFI brightness temperature in our model. We obtain its series expansion to facilitate the calculation of RFI brightness temperature cumulants later on.

   \vspace{-2mm}
\begin{lemma}
    The series expansion of \eqref{eq:MGF_Cluster} is as follows:
    \vspace{-2mm}
    \begin{equation}
        M_{cluster}(t;\; x, \; g)=\sum_{n=0}^{\infty} p_{n}(\lambda_{BS}) \left( g \eta \omega^{\alpha} x^{-\alpha}  \right)^n \frac{t^n}{n!}, \label{eq:MGF_Cluster_Series}
          \vspace{-2mm}
    \end{equation}
    with:
    
      \vspace{-2mm}
    \begin{equation}
        p_n(\lambda_{BS})=\sum\nolimits_{i=0}^{n} S(n,i) \lambda_{BS}^i, \label{eq:p_n}
    \end{equation}
    where $S(n,i)$ is the Stirling number of the second kind.
\end{lemma}
\begin{proof}
    From \cite{hwang2019differential}, we know that:
    \begin{equation}
        \exp \left(v(e^{t}-1) \right) = \sum_{n=0}^{\infty} B_n(v)\frac{t^n}{n!}, \label{eq:exp_v}
    \end{equation}
    where $B_{n}(v)$ is the Bell polynomial of order $n$, and can be expanded as:
    \begin{equation}
        B_{n}(v)=\sum\nolimits_{i=0}^{n}S(n,i)v^i.
    \end{equation}
    By setting $t:=g \eta \omega^{\alpha}x^{-\alpha}t$ and $v=\lambda_{BS}$ in \eqref{eq:exp_v}, we acquire \eqref{eq:MGF_Cluster_Series}.
\end{proof}

\subsection{MGF on SMAP's Main- and Side-lobes}

\subsubsection{SMAP's main-lobe} The RFI brightness temperature $T_{(ml)}$ on SMAP's main-lobe is caused by all the clusters $\mathtt{x} \in \psi_{(ml)}$ in SMAP's main-lobe antenna footprint: i.e., 
\begin{equation}
    T_{(ml)}= \sum\nolimits_{\bold{X}_i \in \psi_{(ml)}} T_{cluster}(\bold{X}_i). \label{eq:T_RFI_ml}
\end{equation}
\begin{note}
    Given that SMAP's antenna footprint is relatively small compared to the distance to the satellite, we assume a uniform distribution, envisioning that all cluster centers within the main-lobe antenna footprint ($\mathtt{x} \in \psi^{(ml)}$) are approximately equidistant from the satellite. Under this assumption, the distance to the satellite for the clusters in the main-lobe is:
    \begin{equation}
    d_{(ml)}= h\cos{(40{\degree})} - \sqrt{r_e^2-h^2\sin^{2}{(40{\degree})}}. 
    \label{eq:ML_Dist_Sat}
\end{equation}
\end{note}

\begin{lemma}
With the assumption of $M \sim \textsf{Pois}(\Lambda)$ equidistant $d_{(ml)}$ clusters from the satellite located in SMAP's main-lobe antenna footprint $\mathcal{B}_{(ml)}$, where $\Lambda=\lambda_c\text{v}^2(\mathcal{B}_{(ml)})$ and $\text{v}^2(\cdot)$ is a Lebesgue measure in $\mathbb{R}^2$, the MGF of $T_{(ml)}$ defined in \eqref{eq:T_RFI_ml} is as follows:
\begin{multline}
    M_{(ml)}(t)= \\ \exp \left(40^{2}\lambda_{c} \left(-1+M_{cluster}(t; \; d_{(ml)}, \; g_{(ml)}) \right) \right), \label{eq:MGF_Main_Lobe}
\end{multline}
where $M_{cluster}(t; \; x,\; g)$ is defined in \eqref{eq:MGF_Cluster}.
\end{lemma}
\begin{proof}
    Refer to Appendix A.
\end{proof}

\subsubsection{SMAP's side-lobe} In this section, we investigate the maximum RFI brightness temperature $T_{(sl)}$ on SMAP's side-lobe defined in \eqref{eq:T_RFI}, which we can rewrite as:
\begin{equation}
    T_{(sl)}=\sum\nolimits_{\bold{X}_i \in \Psi_{(sl)}} T_{cluster}(\bold{X}_i), \label{eq:T_RFI_sl}
\end{equation}
where $T_{cluster}(\bold{X}_i)$ is defined in \eqref{eq:T_cluster}.
\begin{lemma}
    The MGF of \eqref{eq:T_RFI_sl} is as:
    \begin{align}
       & \hspace{-4mm}M_{(sl)}(t)=\nonumber \\
       & \hspace{-5mm} \exp \hspace{-.8mm}{ \left(\hspace{-.6mm}-2\pi\hspace{-.6mm}\left(\frac{r_e}{h}\right)\hspace{-.6mm} \lambda_c\hspace{-.6mm} \int_{d_{min}}^{d_{max}\hspace{-.6mm}}\hspace{-.6mm}\hspace{-.6mm} \left(1-M_{cluster}(t; \;x,\;g_{(sl)}) \right)x\,dx \hspace{-.6mm}\right)} \label{eq:MGF_Side_Lobe}\hspace{-5mm}
    \end{align}
    where $M_{cluster}(t; \;x, \; g)$ is defined in \eqref{eq:MGF_Cluster}, and $d_{min}=h-r_e$ and $d_{max}=\sqrt{h^2 - r^2_e}$.
\end{lemma}
\begin{proof}
Refer to Appendix B.    
\end{proof}

\subsection{Cumulants of RFI Brightness Temperature on SMAP's Main- and Side- Lobes}
Now that we have the MGFs of $T_{(ml)}$ and $T_{(sl)}$, we are able to acquire their cumulants. 
\begin{lemma} The $n$th cumulant of RFI brightness temperature $T_{(ml)}$ on SMAP's main-lobe,  defined in \eqref{eq:T_RFI_ml}, is as follows: 
\begin{equation}
     k_{n}^{(ml)}=40^{2} g_{(ml)}^n \eta^n \omega^{n\alpha} \lambda_cp_n(\lambda_{BS})d_{(ml)}^{-n\alpha}, \label{eq:kn_ml}
\end{equation}
where $p_{n}(\lambda_{BS})$ is defined in \eqref{eq:p_n}.
\end{lemma}
\begin{proof}
    The cumulants of $T_{(ml)}$ can be acquired using its MGF as defined in \eqref{eq:MGF_Main_Lobe} and definition 1.
\end{proof}

\begin{corollary} The expected value of $T_{(ml)}$ defined in \eqref{eq:T_RFI_ml} is:
\begin{equation}
\mathds{E}[T_{(ml)}]=40^2 g_{(ml)} \eta \omega^{\alpha} \lambda_c\lambda_{BS}d_{(ml)}^{-\alpha}.   \label{eq:T_ml_average} 
\end{equation}
\end{corollary}
\begin{proof}
    Based on remark 1, the expected value of $T_{ml}$ is its first cumulant defined in \eqref{eq:kn_ml}.
\end{proof}

\begin{corollary} The variance of $T_{(ml)}$ defined in \eqref{eq:T_RFI_ml} is as:
\begin{equation}
Var[T_{(ml)}]=40^2 g_{(ml)}^2 \eta^2 \omega^{2\alpha} \lambda_c (\lambda_{BS}^2+\lambda_{BS})d_{(ml)}^{-2\alpha}.    \label{eq:T_ml_STD}
\end{equation}
\end{corollary}
\begin{proof}
Based on remark 1, the variance of $T_{(ml)}$ is its \MK{second} cumulant defined in \eqref{eq:kn_ml}.
\end{proof}

\begin{lemma}
    The $n$th cumulants of the RFI brightness temperature $T_{(sl)}$ on SMAP's side-lobe, defined in \eqref{eq:T_RFI_sl}, is as follows:
    \begin{multline}
        k_n^{(sl)}= \\ \frac{2\pi}{2-n\alpha}\left(\frac{r_e}{h}\right) g_{(sl)}^n \eta^n \omega^{n\alpha} \lambda_c p_n(\lambda_{BS}) \left(d_{max}^{2-n\alpha} - d_{min}^{2-n\alpha}\right) \label{eq:kn_sl},
    \end{multline}
    where $p_{n}(\lambda_{BS}$) is defined in \eqref{eq:p_n}.
\end{lemma}
\begin{proof}
The cumulants of $T_{(sl)}$ can be acquired using its MGF as defined in \eqref{eq:MGF_Side_Lobe} and definition 1.
\end{proof}

\begin{corollary} The expected value of $T_{(sl)}$ defined in \eqref{eq:T_RFI_sl} is:
    \begin{multline}
        \mathds{E}[T_{(sl)}]= \\ \frac{2\pi}{2-\alpha}\left(\frac{r_e}{h}\right) g_{(sl)} \eta \omega^{\alpha} \lambda_c \lambda_{BS} \left(d_{max}^{2-\alpha} - d_{min}^{2-\alpha}\right). \label{eq:T_sl_average}
    \end{multline}
\end{corollary}
\begin{proof}
     Based on remark 1, the expected value of $T_{(sl)}$ is its first cumulant defined in \eqref{eq:kn_sl}.
\end{proof}

\begin{corollary} The variance of $T_{(ml)}$ defined in \eqref{eq:T_RFI_ml} is:
    \begin{multline}
        Var[T_{(sl)}]= \\ \frac{2\pi}{2-2\alpha}\left(\frac{r_e}{h}\right) g_{(sl)}^{2} \eta^2 \omega^{2\alpha} \lambda_c (\lambda_{BS}^2+\lambda_{BS}) \left(d_{max}^{2-2\alpha} - d_{min}^{2-2\alpha}\right). \label{eq:T_sl_STD}
    \end{multline}
\end{corollary}
\begin{proof}
    Based on remark 1, the variance of $T_{(sl)}$ is its \MK{second} cumulant defined in \eqref{eq:kn_sl}.
\end{proof}

% \vspace{-10 mm}
\section{Numerical Results} \label{sec:Results}
In this section, we evaluate the RFI brightness temperature statistics at both SMAP's main- and side-lobes based on the analysis in the previous sections. The simulation parameters are given in Table \ref{tab:SG_Params}. Based on this table, an average of one cluster exists in every $10000$ $km^2$ and, on average, $2500$ BS clusters exist in the area exposed to the satellite. We also set SMAP's side-lobe gain to a conservative $-55$ dB.

\begin{table}[]
\centering
\caption{LIST OF CONFIGURATION PARAMETERS FOR PERFORMANCE EVALUATION}
\label{tab:SG_Params}
\begin{tabular}{|c|c|}
\hline
\textbf{Element} & \textbf{Value}                                                                              \\ \hline
Intensity of clusters ($\lambda_C$)      & \begin{tabular}[c]{@{}c@{}}$1$ cluster (large city) \\ every $10000$ $km^{2}$\end{tabular}  \\ \hline
Intensity of active BS ($\lambda_{BS}$)   & \begin{tabular}[c]{@{}c@{}}$50$, $100$, $200$\\ BSs per cluster\end{tabular} \\ \hline
Path loss exponent ($\alpha$)         & {(}2 , 2.2{]}                                                                               \\ \hline
BS transmission power ($p_{tx}$)                & $3.5$ Watts per BS    \\ \hline
Boltzmann's constant ($k_b$) & \begin{tabular}[c]{@{}c@{}} $1.380649 \times 10^{-23}$\\ m$^{2}$kg s$^{-2}$K$^{-1}$ \end{tabular} \\ \hline
Light speed ($c$) & $300,000$ $km$ $s^{-1}$ \\ \hline
SMAP's main-lobe antenna gain ($g_{(ml)}$)            & $0$ dB                                                                                    \\ \hline
SMAP's side-lobe antenna gain ($g_{(sl)}$)            & $-55$ dB                                                                                    \\ \hline
Central Carrier frequency of BS ($f$)                & $1.413$ GHz                                                                                                                                                           \\ \hline
BS transmission bandwidth ($\beta$)                & $24$ MHz                                                                                    \\ \hline
Earth radius ($r_e$)                & $6371$ km                                                                                   \\ \hline
SMAP's distance to Earth's center ($h$)                & $7056$ km                                                                                   \\ \hline
Min. possible distance to SMAP ($d_{\min}$)        & $h-r=685$ km                                                                                    \\ \hline
Max. possible distance to SMAP ($d_{\max}$)        & $\sqrt{h^2-r^2}=3032.7$ km                                                                            \\ \hline
\end{tabular}
\end{table}

\begin{figure*}
  \centering
    %\begin{minipage}{.5\linewidth}
    \centering
    \begin{subfigure}{0.44\textwidth}
      \centering
      \includegraphics[width=\linewidth]{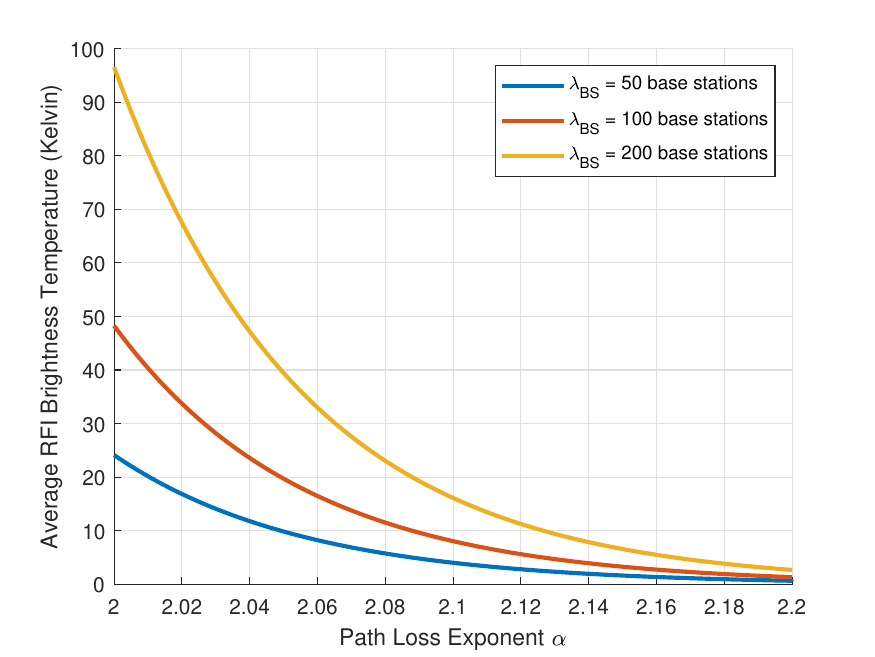}
      \caption{Main-lobe average RFI brightness temperature.}
      \label{subfig:1}
    \end{subfigure}
    % \vspace{2 mm}
    \begin{subfigure}{0.44\textwidth}
      \centering
      \includegraphics[width=\linewidth]{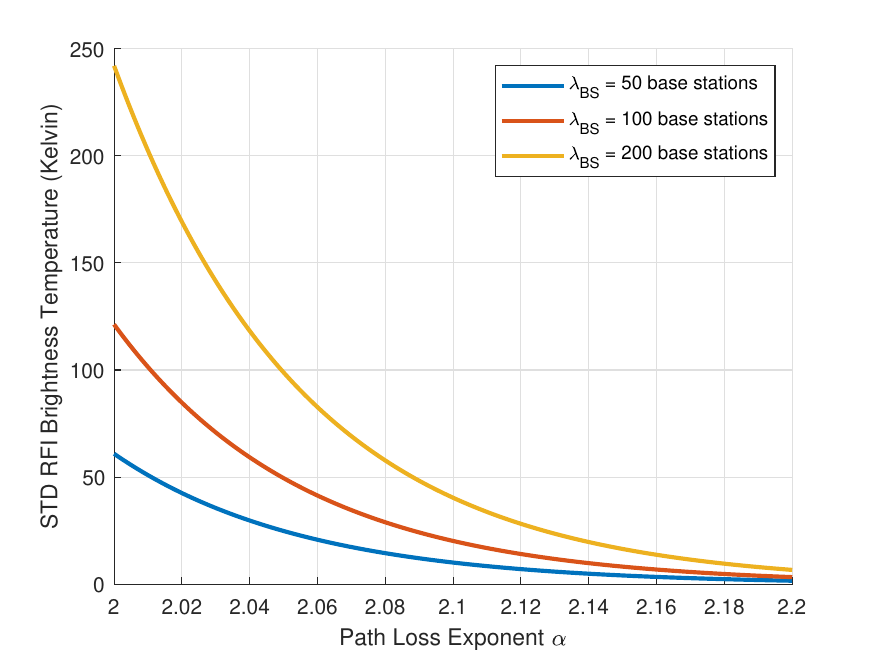}
      \caption{Main-lobe STD of RFI brightness temperature.}
      \label{subfig:2}
    \end{subfigure}
    
  %\end{minipage}%
  %\begin{minipage}{.5\linewidth}
    \centering
    \begin{subfigure}{0.44\textwidth}
      \centering
      \includegraphics[width=\linewidth]{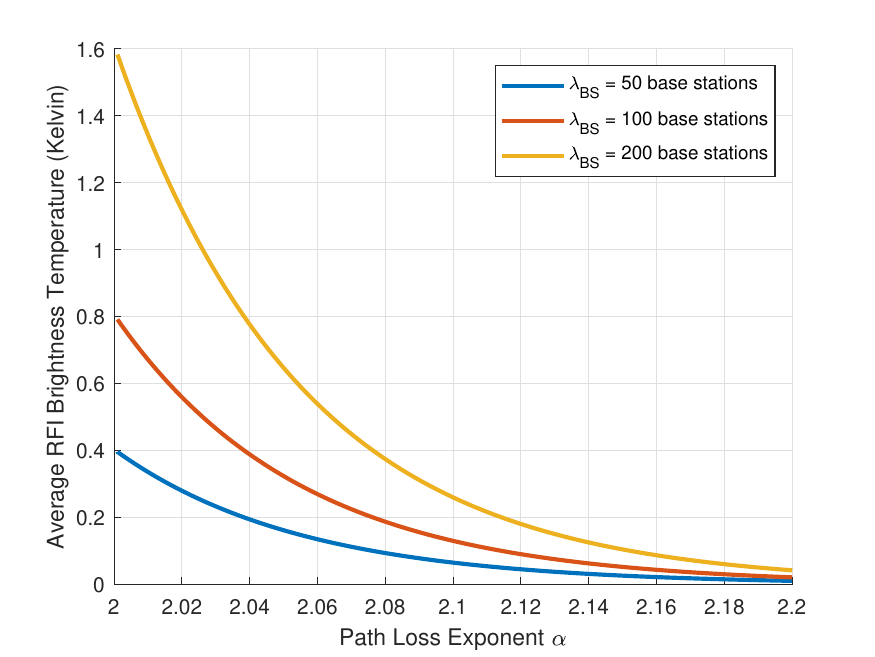}
      \caption{Side-lobe average RFI brightness temperature}
      \label{subfig:3}
    \end{subfigure}
    % \vspace{3 mm}
    \begin{subfigure}{0.44\textwidth}
      \centering
      \includegraphics[width=\linewidth]{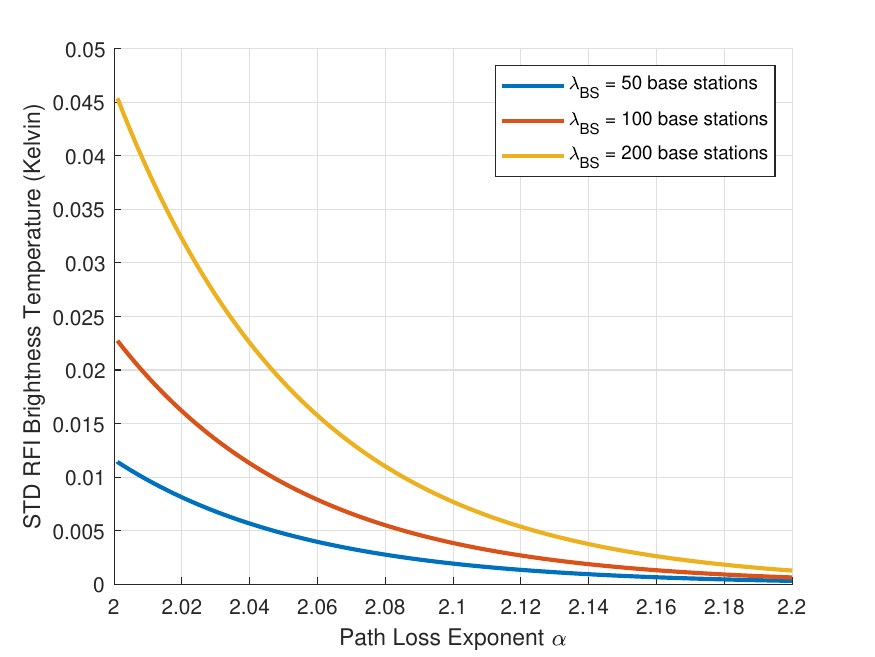}
      \caption{Side-lobe STD of RFI brightness temperature.}
      \label{subfig:4}
    \end{subfigure}
  %\end{minipage}
  \caption{Average and Standard Deviation of RFI brightness temperature on SMAP's main- and side-lobes for average $\lambda_{BS}=50, \;100,\; 200$ active BSs in a cluster. On average $125,\; 250$ and $500$ thousand base stations are exposed to SMAP's side-lobe, while only, on average, $7, \; 14$ and $28$ base stations are exposed to SMAP's main-lobe.}
  \label{fig:Exc-Prob}
\end{figure*}

% Due to limited space, we focus on evaluating the average and standard deviation of RFI brightness temperature in both the main and side lobes of SMAP. In doing so, we use the average RFI to gauge whether the brightness temperature falls within an acceptable error range. Additionally, we analyze the standard deviation of RFI to assess the distribution's tail and the likelihood of RFI brightness temperature surpassing an acceptable threshold. To assess average and STD of RFI brightness temperature on SMAP's main- and side-lobes, we use \eqref{eq:T_ml_average}, \eqref{eq:T_ml_STD} and \eqref{eq:T_sl_average}, \eqref{eq:T_sl_STD}, respectively. Figure (\ref{fig:Exc-Prob}), depicts the average and STD of RFI brightness temperature on SMAP's main- and side-lobes.

% According to Figure \ref{subfig:1}, the average RFI brightness temperature imposed on SMAP's main-lobe is in the range of $20$ to $100$ Kelvin, which by far exceeds the acceptable $\tau=1.3K$ for SMAP. Also according to Figure \ref{subfig:2}, the STD of RFI brightness temperature on SMAP's main-lobe is also very high and in the range of hundreds, showing the extreme randomness of RFI brightness temperature on SMAP's main-lobe. The important point is that, on average, $7$ to $28$ base stations are creating such RFI brightness temperature characteristics.

Due to space limitations, we only assess the average and standard deviation (STD) of RFI brightness temperature at SMAP's main- and side-lobes. We use these metrics to determine if the values fall within an acceptable error range. Figure \ref{fig:Exc-Prob} illustrates the findings. In Figure \ref{subfig:1}, the average RFI brightness temperature at SMAP's main-lobe exceeds the acceptable threshold $\tau=1.3K$, ranging from $20$ to $100$ Kelvin. Figure \ref{subfig:2} shows a high standard deviation in the range of hundreds, indicating extreme uncertainty in RFI. Notably, $7$ to $28$ base stations, on average, contribute to these RFI characteristics. 

Now we focus on RFI brightness temperature on SMAP's side-lobe depicted in Figures \ref{subfig:3} and \ref{subfig:4}. As observable in figure \ref{subfig:3}, for clusters with an average of $50$ and $100$ active BSs, the average RFI brightness temperature falls within the acceptable threshold of $1.3K$, while for clusters of average $200$ active BSs, the average RFI brightness temperature slightly exceeds the limit. Also, from Figure \ref{subfig:4}, we see extremely low STD values for RFI brightness on SMAP's side-lobe, which is an indicator of extremely low uncertainty in RFI.

\section{Conclusion} \label{sec:conclusion}
% In this paper, we developed an innovative method based on cluster processeses in stochastic geometry to assess the RFI brightness temperature on passive remote sensing satellites, induced by an imaginary extensive terrestrial cellular network. Specifically, we crafted our model on NASA's SMAP remote sensing satellite, operating actively within the confined L-band $1400$ to $1427$ MHz. Intriguingly, our exploration revealed that the remarkably low side-lobe gains of SMAP's antenna open the door for a multitude of active co-channel cellular base stations, all the while safeguarding the precision of SMAP's measurements.

% In this paper, we introduced an innovative method using cluster processes in stochastic geometry to evaluate RFI brightness temperature on passive remote sensing satellites. Our model is tailored to NASA's SMAP satellite, passively operating in the restricted L-band $1400$ to $1427$ MHz. Interestingly, our findings show that SMAP's antenna's remarkably low side-lobe gains allow for numerous active co-channel cellular base stations without compromising measurement precision.

In this paper, we introduced an innovative method using cluster processes in stochastic geometry to evaluate RFI brightness temperature on passive remote sensing satellites. Our model is tailored to NASA's SMAP satellite, passively operating in the restricted L-band $1400$ to $1427$ MHz. Interestingly, our findings show that\MK{, while avoiding active base stations in main-lobe antenna footprint through satellite positioning,} SMAP's antenna's remarkably low side-lobe gains allow for numerous active co-channel cellular base stations without compromising measurement precision.

\appendices
\input{Appendix_MGF_RFI_mL}
\input{Appendix_MGF_RFI_SL}

\bibliographystyle{IEEEtran}
\bibliography{refs}

\end{document}

%% file: Appendix_MGF_RFI_ml.tex
\section{Proof to Lemma 3}
Since SMAP's main-lobe antenna footprint is roughly $40^2$ $km^2$, we note that $\Lambda=40^2\lambda_c$. Thus, conditioned on $M$ we note that:
\begin{equation}
    M_{(ml)}(t)=\mathds{E}_{M} \left[ \mathds{E} \left[\exp\left(t\sum\nolimits_{M}T_{cluster}(\mathtt{x})\right) \right] \right]. \label{eq:MGF_Main_Lobe_proof}
\end{equation}
The inner $\mathds{E}[\cdot]$ is the MGF of the sum of $M$ clusters defined in \eqref{eq:T_cluster}, and based on \eqref{eq:MGF_Cluster}, can be expressed as:
\begin{multline}
    \mathds{E} \left[\exp\left(t\sum\nolimits_{m}T_{cluster}(\mathtt{x})\right) \right]= \\\left( M_{cluster}(t;\; d_{(ml)}, \;g_{(ml)}) \right)^m.
\end{multline}
Accordingly, \eqref{eq:MGF_Main_Lobe_proof} can be expanded as:
\begin{equation}
    M_{(ml)}(t)=\sum_{m=0}^{\infty}\frac{e^{-\Lambda}\Lambda^m}{m!} \left( M_{cluster}(t;\; d_{(ml)}, \;g_{(ml)}) \right)^m,
\end{equation}
which \MK{is equivalent to} \eqref{eq:MGF_Main_Lobe}.

%% file: Appendix_MGF_RFI_SL.tex
\section{Proof to Lemma 4} \label{proof:MGF_RFI_SL}
For ease of notation, in \eqref{eq:T_RFI_sl} we define $T_{\bold{X}_i}:=T_{cluster}(\bold{X}_i)$. Accordingly, the MGF of \eqref{eq:T_RFI_sl} is as:
\begin{align}
    M_{(sl)}(t) &=\mathds{E}\left[ e^{tT_{(sl)}}\right] \notag \\
    & = \mathds{E}_{\Psi,\{ T_{\bold{X}_i} \}} \left[
    \prod\nolimits_{\bold{X}_i \in \Psi^{(sl)}} e^{tT_{\bold{X}_i}}
    \right].
\end{align}
Due to the initial assumption of the independence of clusters, we move the expectation with respect to $\{ T_{\bold{X}_i} \}$ inside the product as: 
\begin{equation}
    \mathds{E}_{\Psi_{(sl)}} \left[
         \prod\nolimits_{\bold{X}_i \in \Psi_{(sl)}}
         \mathds{E} \left[ 
           e^{tT_{\bold{X}_i}}
         \right]
    \right],
\end{equation}
% \\---------------------------------------------------------------------
% \\---------------------------------------------------------------------
% \\---------------------------------------------------------------------\\
% The MGF of \eqref{eq:T_RFI_sl} is as:
% \begin{align}
%     M^{(sl)}(t) &=\mathds{E}\left[ e^{tT^{(sl)}}\right] \notag \\
%     &=\mathds{E} \left[ \exp{ \left( t \left( \sum_{\bold{X}_i \in \Psi^{(sl)} } T_{cluster}^{(sl)}(\bold{X}_i) \right) \right) }    \right]
% \end{align}
% for the ease of notation we define $T_{\bold{X}_i}^{(sl)}:=T_{cluster}^{(sl)}(\bold{X}_i)$. Thus \eqref{}, expands as:
% \begin{equation}
%     \mathds{E}_{\Psi,\{ T_{\bold{X}_i}^{(sl)} \}} \left[
%     \prod_{\bold{X}_i \in \Psi^{(sl)}} e^{tT_{\bold{X}_i}^{(sl)}}
%     \right]
% \end{equation}
% With the assumption of independence of clusters, \eqref{} can be simplified as:
% \begin{equation}
%     \mathds{E}_{\Psi} \left[
%          \prod_{\bold{X}_i \in \Psi^{(sl)}}
%          \mathds{E} \left[ 
%            e^{tT_{\bold{X}_i}^{(sl)}}
%          \right]
%     \right]
% \end{equation}
which is the \textit{Probability Generating Functional} (PGFL) \cite{sg1} of $g(\mathtt{x})=  \mathds{E}\left[ e^{tT_{\mathtt{x}}} \right]$  over the set $\mathcal{B}_{(sl)}$ and can be written as:
\begin{multline}
\small
    \mathcal{P}_{\Psi_{(sl)}}(g) = \mathds{E}_{\Psi_{(sl)}} \left[ \prod_{\mathrm{X}_{i} \in \Psi_{(sl)}} g(\mathrm{X}_{i}) \right]
    \\ = \exp \left( 
           -\int_{\mathcal{B}_{(sl)}} \left( 1-g(\mathtt{x}) \right) \, \Lambda(d\mathtt{x})
    \right), \label{eq:PGFL}                     
\end{multline}
where, based on Fig. (\ref{fig:smap_exposed}), in spherical coordinates for equidistant points to SMAP for a PPP with intensity $\lambda_c$:
\begin{equation}
    \Lambda \left( d \mathtt{x} \right) = 2 \pi r^2 \lambda_c \sin(\theta) \; d\theta.   \label{eq:Lambda_dX}
\end{equation}
We note that $g(\mathtt{x})$ is the MGF of RFI brightness temperature of a cluster at point $\mathtt{x} \in \mathcal{B}_{(sl)}$ as in \eqref{eq:MGF_Cluster}, which we note here with $g(\mathtt{x})=M_{cluster}(t; \; \lVert \mathtt{x}-\mathtt{h} \lVert, \;g_{(sl)} )$, where $x=\lVert \mathtt{x}-\mathtt{h} \lVert$ is the distance to SMAP. From Fig. (\ref{fig:smap_exposed}), and using the law of cosines, we note that:
\begin{equation}
    x= \left( r_e^2 + h^2 -2hr_e \cos (\theta) \right)^{\frac{1}{2}}, 
\end{equation}
with:
\begin{align}
    &dx = hr_e \sin(\theta) \left( r_e^2 + h^2 -2hr_e \cos (\theta) \right)^{-\frac{1}{2}}\; d\theta \notag \\
    \Rightarrow x\;&dx= hr_e \sin(\theta)\; d\mathtt{\theta}. \label{eq:xdx}
\end{align}
Comparing \eqref{eq:Lambda_dX} and \eqref{eq:xdx}, we note that:
\begin{equation}
    \Lambda \left( d \mathtt{x} \right) = 2 \pi \left( \frac{r_e}{h} \right) \lambda_cx \;dx, \label{eq:lambda_to_x}
\end{equation}
where $x$ is in the range of $d_{max}$ and $d_{min}$. By substituting \eqref{eq:lambda_to_x} in \eqref{eq:PGFL}, we will have \eqref{eq:MGF_Side_Lobe}.